\documentclass[letterpaper, 10 pt, conference]{ieeeconf}  

\usepackage{savesym}  
\savesymbol{proof}  
\savesymbol{endproof}  

\usepackage{amsmath,amsfonts,amssymb}
\usepackage{amsthm}
\usepackage{graphicx}
\usepackage{subeqnarray}
\usepackage{soul}
\usepackage{algpseudocode}
\usepackage{algorithm}
\usepackage{verbatim}
\usepackage{drum}
\usepackage{rpmath}
\usepackage{url}

\IEEEoverridecommandlockouts                              

\catcode`\^^M=10    

\usepackage{array}
\newcolumntype{M}[1]{>{\centering\arraybackslash}m{#1}}
\newcolumntype{L}[1]{>{\raggedleft\arraybackslash}m{#1}}
\newcolumntype{R}[1]{>{\raggedright\arraybackslash}m{#1}}
\newcolumntype{N}{@{}m{0pt}@{}}

\newsavebox{\measurebox}

\newcommand{\RDot}{\dot{R}_{A B}}
\newcommand{\RDotNHat}{\dot{\point{R}}_{A B \nhat}}
\newcommand{\RDotPerp}{\dot{\point{R}}_{A B \perp}}

\newtheorem{thm}{Theorem}

\title{\LARGE \bf A pressure field model for fast, robust approximation of net contact force and moment between nominally rigid objects}

\author{Ryan Elandt$^{1,2}$, Evan Drumwright$^1$, Michael Sherman$^1$, Andy Ruina$^2$
\thanks{$^1$Toyota Research Institute, Los Altos, CA 94022}
\thanks{$^2$Cornell University, Mechanical  Engineering, Ithaca, NY 14850}
}

\begin{document}

\maketitle
\pagestyle{empty} 

\noindent
\begin{abstract}
We introduce Pressure Field Contact (PFC), an approximate model for predicting the contact surface, pressure distribution, and net contact wrench between nominally rigid objects. 
PFC combines and generalizes two ideas: a bed of springs (an `elastic foundation') and hydrostatic pressure. Continuous pressure fields are computed offline for the interior of each nominally rigid object.
Unlike hydrostatics or elastic foundations, the pressure fields need not satisfy mechanical equilibrium conditions.
When two objects nominally overlap, a contact surface is defined where the two pressure fields are equal.
This static pressure is supplemented with a dissipative rate-dependent pressure and friction to determine tractions on the contact surface.
The contact wrench between pairs of objects is an integral of traction contributions over this surface.
PFC evaluates much faster than elasticity-theory models, while showing the essential trends of force, moment, and stiffness increase with contact load.
It yields continuous wrenches even for non-convex objects and coarse meshes.
The method shows promise as sufficiently fast, accurate, physical, and robust for robotics applications including motion and tactile sensor simulation, controller learning and synthesis, state estimation, and design-in-simulation. 
\end{abstract}

\section{INTRODUCTION}
For much of robotics, simulations based on motions of rigid objects are likely accurate enough.
A main simulation weakness, however, is the modeling of the forces and deformations of contact between these objects.
Contact models tend to have at least one of these problems: slow computation, non-physical artifacts, severe inaccuracies, or applicability only to very simple geometry.
Our primary goal with this work is improving simulations through robust contact modeling that effectively balances physical accuracy and running time, toward developing and testing robot controllers.

\par
\textbf{Rigid objects.}
When objects are sufficiently stiff, motion of material points due to deformation is small compared to average translation and rotation.
For such motions, energy, force, and momentum equations are well approximated by those of rigid objects (`rigid bodies').
Hence, the centuries-long tradition of modeling many systems --- especially mechanisms made largely of hard materials like metal, stone, wood, and bone --- using `rigid-object' mechanics.
Errors from this assumption are typically on the order of the strains, well under 1\%.

\par
\textbf{Contact forces.} Motion comes from forces, some of which are generated when objects touch.
Forces from assumed-to-be rigid and workless contacts can be algebraically eliminated from the dynamics or, if needed, solved for as part of the dynamics.
This `motion constraint'~\cite{featherstone:2008} approach to contact forces, neglecting contact deformation, is usually considered accurate enough for, \eg lubricated pin and sliding joints.

\par
\textbf{Contact deformation can't always be neglected.}
In fact, however, contact forces are resultants of stresses (or tractions) caused by deformation (strain) in the otherwise nominally-rigid objects \cite{Chatterjee:1998a}.
In some situations, the formulation of contact forces as motion constraints is not sufficiently accurate.  
Situations where some treatment of deformation is needed include: configurations with indeterminate contact forces (\eg the support forces on the feet of a four-legged chair); frictional contacts, where the resistance to rotation about the nominal contact normal depends on the deformation-dependent area of contact (\eg resistance to rotation in a pinched grip); and simultaneous collisions (for which algebraic collision laws are arbitrary or inaccurate \cite{Chatterjee:1999}).

\par
\textbf{Rigid objects with contact deformation.}
Rigid-object models are commonly used because rigid-object computations are \textit{so much} faster than, say, models where objects are modeled using elasticity, and because for many problems, the loss of accuracy by neglecting bulk deformation is small compared to other modeling errors (\eg in mass distribution, geometry, or material properties).
Then, locally violating the rigidity assumption, contact forces (normal forces, friction, and collisional impulses) are found using simplified deformation models including \1 point contact with collisional restitution and friction, \2 discrete springs/dampers, \3 Hertz contact, and \4 finite-elements models of contact regions \cite{Chatterjee:1998a}.

\par
State of the art contact approaches generally lie at the ends of the speed/accuracy spectrum.
Rigid point-contact based approaches are fast, but are non-smooth, hard to apply to objects with arbitrary geometries, and inefficient at computing all but coarse solutions. And, point contact models miss area-dependent phenomena especially the net contact moment, \eg scrubbing torques and rolling (or tipping) resistance.
At the other extreme, finite element (FEM) deformation models can be precise (perhaps inappropriately precise if geometry is uncertain) but generally require orders of magnitude more time to simulate than rigid-contact methods.
Both extremes have their place: simplicity is needed to simulate thousands of particles (\eg granular flow), while fidelity is needed to understand single contacts in detail (\eg car tire modeling).
Often, neither extreme is both sufficiently fast and sufficiently accurate for scenarios that include collisions, rolling resistance, and sliding (or resistance to sliding) (\eg \cite{Brogliato:2002}, p. 107).

\par
\textbf{Key desired contact-model features.}
Our focus is situations (\eg manipulation and locomotion) where these contact area-related torques can be important. We seek a contact model that captures area-dependent phenomena without the computational cost of FEM.

\par
For effective learning and optimization we need continuous gradients, meaning a contact law must predict forces that are continuous functions of state. 
For computational expedience, we want to avoid solving complementarity problems (as explained in \eg \cite{featherstone:2008}).
Similarly, we want to minimize the number of extra state variables associated with a contact patch (\eg FEM approaches  parameterize contact deformation region using many state deformation variables).
For general usage in robotics, the approach must work for arbitrary shapes (curved or with corners, convex, or non-convex) and produce physically-reasonable results for coarsely-represented geometry.

\section{Background}
Before describing the PFC model, we further review some alternative approaches.

\subsection{Dynamic, PDE-based wave-propagation models} \label{sec:pde_lcp}
Dynamic FEM approaches approximate the continuum dynamics solution of the full deformation field partial differential equations (PDE), including contact region stresses (and thus forces).
FEM uses multidimensional piecewise polynomials (\ie isoparametric elements) to model spatial variation in material state (position and velocity throughout the material).
FEM converges (in principle) to the exact PDE solution for sufficiently fine meshes and small time-steps.
The grid size must be smaller than the feature size and the time-step must be smaller than the wave transit time to resolve elastic waves.
A high-fidelity elastodynamic calculation is typically inappropriate for robotics-like simulations because internal vibrations are typically only minimally excited.
Quasi-static models provide similarly useful information with much less computation when robots are truly soft, the application targeted in~\cite{Duriez:2017}.

\subsection{Quasi-static, only-local deformation models}

As Hertz noted in the context of elastic objects whose convex shapes were approximated as conic sections \cite{Hertz:1895}, it is often reasonable to assume that the deforming contact region is small and that the object is rigid outside this deformable region.
Thus, the characteristic time in solutions is associated with the whole-object mass interacting with a contact spring whose oscillation period is much longer than the wave transit time across the contact region.
This simplification, to a rigid object with contact forces mediated by a quasi-static contact region, can be generalized to arbitrary contact shapes and to non-linear deformation laws even if Hertz contact is not appropriate. However, elastodynamics cannot be neglected for collisions of high aspect ratios objects such as rods and shells, in which wave transit times are comparable to contact durations. PFC does not attempt capture these phenomena either.

\subsection{Overlap volume methods}
These methods use overlap volume or submerged area to compute configuration-dependent normal forces.
They \1 determine a contact patch and average normal for an overlap region and \2 apply force in this direction.
Methods include \cite{Hasegawa:2003,Luo:2006}, which assume that vector area defined by the boundary of the overlap region determines patch direction; \cite{Hasegawa:2003} uses a `reverse spring foundation' for distributed force; \cite{Luo:2006} projects 3D overlap volume onto the patch and applies a lumped force at the patch centroid proportional to its area; \cite{Gonthier:2005}, which uses the overlap volume's inertia to determine a patch normal and applies a force at the overlap volume's centroid; and \cite{Wakisaka:2017}, which combines a volumetric intersection computation with kinematic constraints. 
These methods rely on a single representative patch normal, so they do not extend well to non-convex contact (\eg a peg in a hole).
Related to the overlap volume approach, are hydrostatic and elastic foundation approaches, discussed further below.

\subsection{Gaming/animation physics approaches to contact}
Most current gaming and animation physics approaches use point-contact constraint forces in the context of \emph{velocity-stepping}~\cite{Brogliato:2002}.
Some libraries targeted to engineering applications (\eg Algoryx, MuJoCo) also use this strategy.
Velocity-stepping approaches are first-order accurate~\cite{Anitescu:2004a} \emph{at best}.
These impulsive approaches provide neither error estimates nor error control.
That said, the approach of~\cite{Guendelman:2003}, while requiring considerable hand tuning to avoid context-specific artifacts, can simulate thousands of non-convex contacting objects.

\section{The pressure field model} \label{section:approach}

The PFC model attempts to capture general trends associated with continuum mechanics calculations, including that stiffness and contact area increase with load, but with a  simpler and faster calculation.

\par
At a high level, the model concerns the non-dissipative repulsion of two nominally rigid objects.

\par\noindent
\textbf{\1 A single net wrench} is calculated for each pair of contacting nominally-rigid objects, even if, for non-convex objects, the contacting region is not connected. 

\par\noindent
\textbf{\2 Pressure field.} Each object uses an immutable object-fixed `virtual pressure field' $p_0$.  
This is not a real pressure, and need not satisfy any mechanical equilibrium conditions (\eg hydrostatics).
This scalar field is zero on the surface and increases with depth.
The pressure field can be generated to increase linearly with depth, or could increase more quickly if bounded interpenetration is desired.  For purely rigid objects, pressure fields jump from  zero  to infinity  over  an  arbitrarily  small  distance  from  the  boundary. 
Finally, the `pressure' field may better be considered a \textit{potential} pressure field, in that it only manifests as an actual pressure in a particular context (discussed below).

\begin{figure}[t]
\centering
	\includegraphics[width=0.5\textwidth]{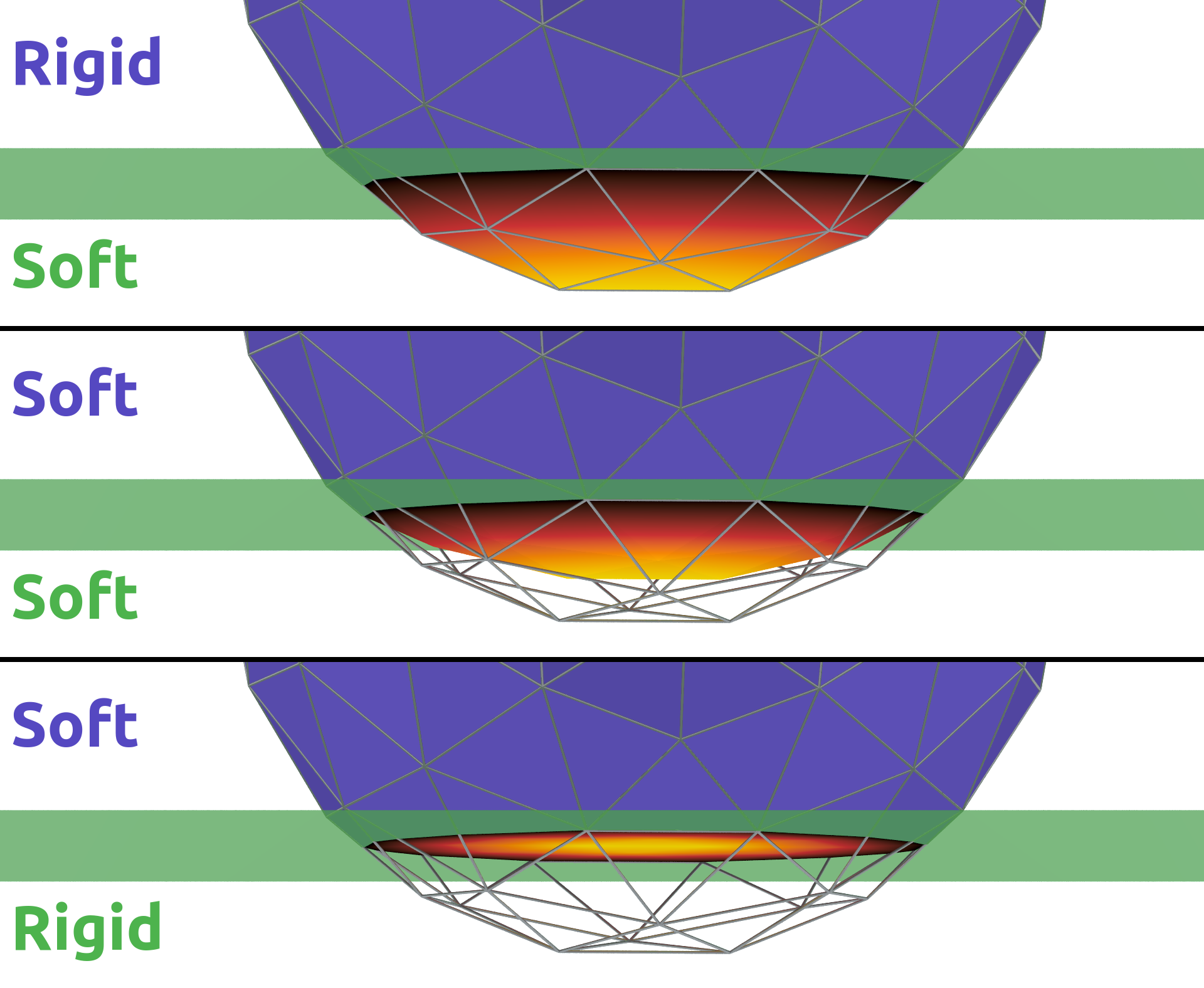}
\caption{\footnotesize\small The contact surfaces (the locus of points at which the pressure field values are equal for both objects) between a sphere and a halfspace for three relative compliances. Lighter areas on the contact surface indicate higher pressure. Neither surface needs to be planar or convex. }
\label{fig:contact-patch-with-heat-map}
\end{figure}

\par\noindent
\textbf{\3 Contact surface.}
Two overlapping objects $A$ and $B$ have pressure fields $p_{0_A}$ and $p_{0_B}$.
The equal-pressure surface ${\mathcal S}^\cap$ (see an example in Figure~\ref{fig:contact-patch-with-heat-map})
is defined by $p_{0_A}=p_{0_B}$ in the overlap region.
This surface is not necessarily connected for non-convex shapes (\ie the surface might be two or more disconnected surfaces).
This iso-surface is straightforward to compute if one judiciously selects the spatial data structures as we  do in~\S\ref{section:implementation}. 

\par\noindent
\textbf{\4 Net contact wrench.}
In the pressure field contact model, \textit{the net contact wrench is calculated as the integral of pressure effects over the contact surface.}
These wrenches tend  to separate overlapping bodies, as shown in Figure~\ref{fig:iso_laplace}, where a compliant object pushes against a rigid rectangular object (and vice versa).

\par\noindent
\textbf{\5 Dissipation and friction.} The model is extended in~\S\ref{section:discussion} to include friction, through shear tractions on the contact surface, and bulk dissipation, through inclusion of a normal traction varying with deformation rate. 

\par \noindent
\textbf{\6 Existence of a strain energy.}
A key feature of the pressure-field contact model is that it is conservative. That is, a contact interaction has an elastic potential energy. Through the course of contact gain and loss, the potential energy goes from zero to zero and the pressure forces do no net work on the pair of objects. The pressure field model is conservative for any scalar pressure field;  the pressure field need not satisfy any mechanics nor equilibrium conditions. 
The strain energy $U_{p_{0_A}}$ associated with $V^{\cap}_A$ is (see~\S\ref{s:pot_energy_proof}): 

\begin{equation}\label{e:int_P_dV}
U_{p_{0_A}} = \int_{V^\cap_A} p_{0_A}(x,y,z) ~dx ~dy ~dz
\end{equation}

where the volume of $A$ displaced is $V^{\cap}_A$.
The strain energy associated with the mutual displacement of the pressure fields in $A$ and $B$ is $U_{p_{0_A}} + U_{p_{0_B}}$.

\section{Relation to other models}
The pressure-field contact model is most reminiscent of elastic foundation, hydrostatic, and volume overlap models.
It shares with all of them that the contact wrench, at least before friction and rate effects are taken into account, is  determined by the instantaneous relative pose of the contacting pair.
In some situations, the pressure field model agrees, or nearly agrees, with one of these models.
The main distinguishing features of the pressure field model are \1 its output varies continuously with changing geometry, \2 it is conservative (associated with a contact potential energy), \3 it does not depend on small relative angles of the contacting surfaces, nor even simply connected regions of overlap, and \4 it can be tuned to make contact stiffness penetration dependent.

\par
\textbf{Relation to hydrostatics.}
The pressure in fluids at equilibrium increases linearly with depth.
Boats float due to this pressure acting on their hulls.
The PFC model is identical to this hydrostatic model in the case where one of our objects has a flat surface and we choose a pressure field that increases linearly with distance below the surface, and where the penetrating object is rigid (\ie has infinite pressure just inside its outer surface); see Figure~\ref{fig:hydro_vs_spring}.
In this case, the equal-pressure surface ${\mathcal S}^\cap$ is thus the outer boundary of the rigid object.  

\par
For non-flat objects, PFC generalizes hydrostatics. We keep the pressure field but not tie it to equilibrium of a real or imagined fluid.
We can select pressure fields where contact stiffness varies with penetration depth in a manner we choose. 
Also, both objects have pressure fields (which express each object's compliance), not just one. Thus the interaction compliance depends on the compliance of both objects.

\begin{figure}[t]
\centering
\includegraphics[width=0.6\linewidth]{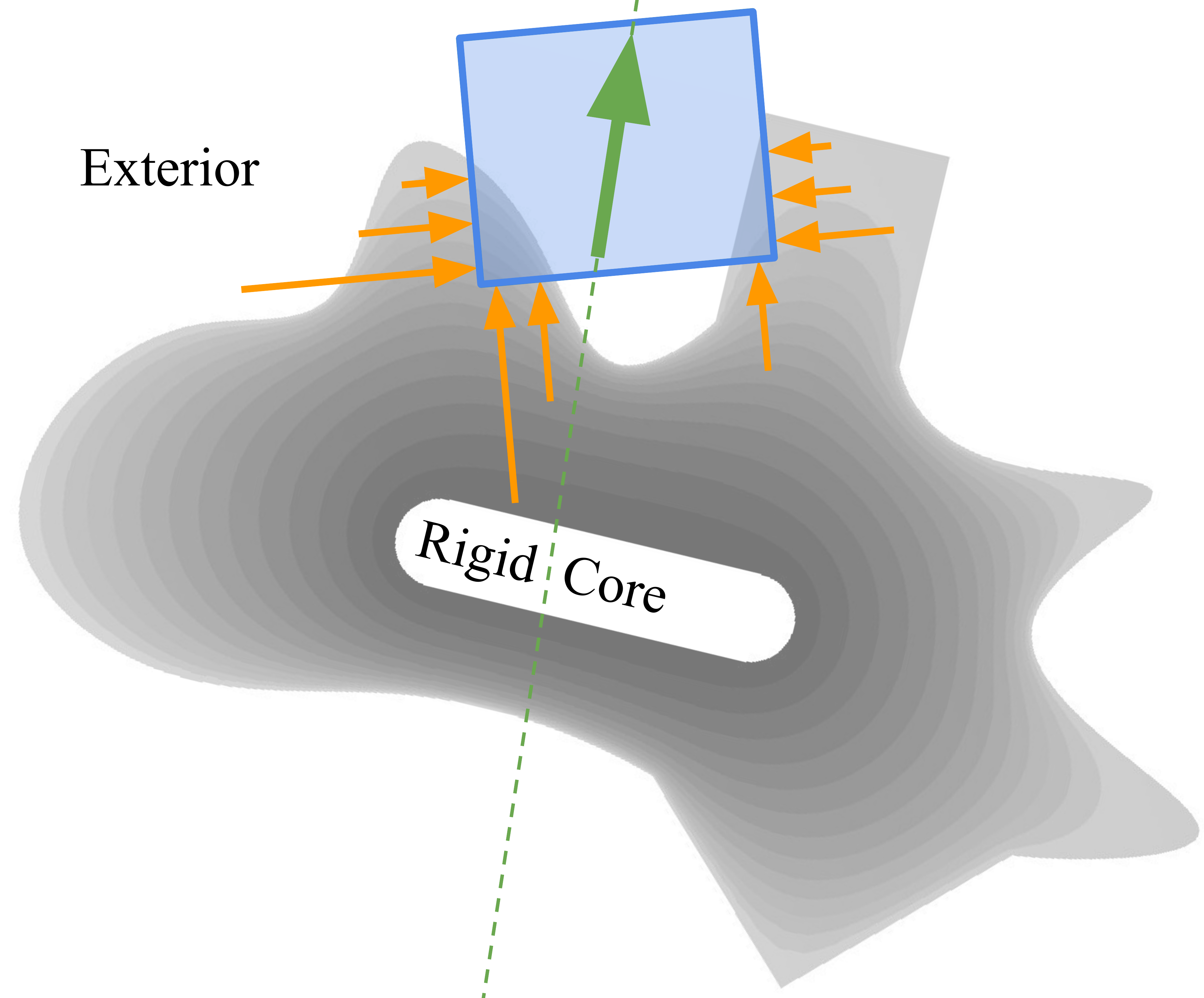}
\caption{\footnotesize\small
    The pressure field for a non-trivial 2D geometry is visualized with gray isocontours.
    A rigid box (blue) penetrates this compliant geometry.
    Orange arrows show the normal stress the compliant object exerts on the rigid geometry.
    The direction of net force is shown in green.
}
\label{fig:iso_laplace}
\end{figure}

\par
\textbf{Relation to elastic foundations.}
The contact of gently curved objects and an elastic layer is equivalent to contact with a continuous bed of springs (for long wavelengths), with normal traction proportional to deflection.
If surfaces of the penetrating objects are not parallel, the elastic foundation model gives traction in the penetration direction while the hydrostatic model gives traction orthogonal to the surface (see Figure~\ref{fig:hydro_vs_spring}).
Even for the contact of a uniform thickness compliant layer and a rigid object, for large relative surface angles, both continuous and discrete elastic foundation models are generally not conservative because the spring ends are generally not anchored (\ie associated with particular material points).
Thus such models, although nominally made from springs, can add or subtract energy from a system.
As noted, the pressure field model is conservative.

\begin{figure}[t!]
\centering
\includegraphics[width=0.49\linewidth]{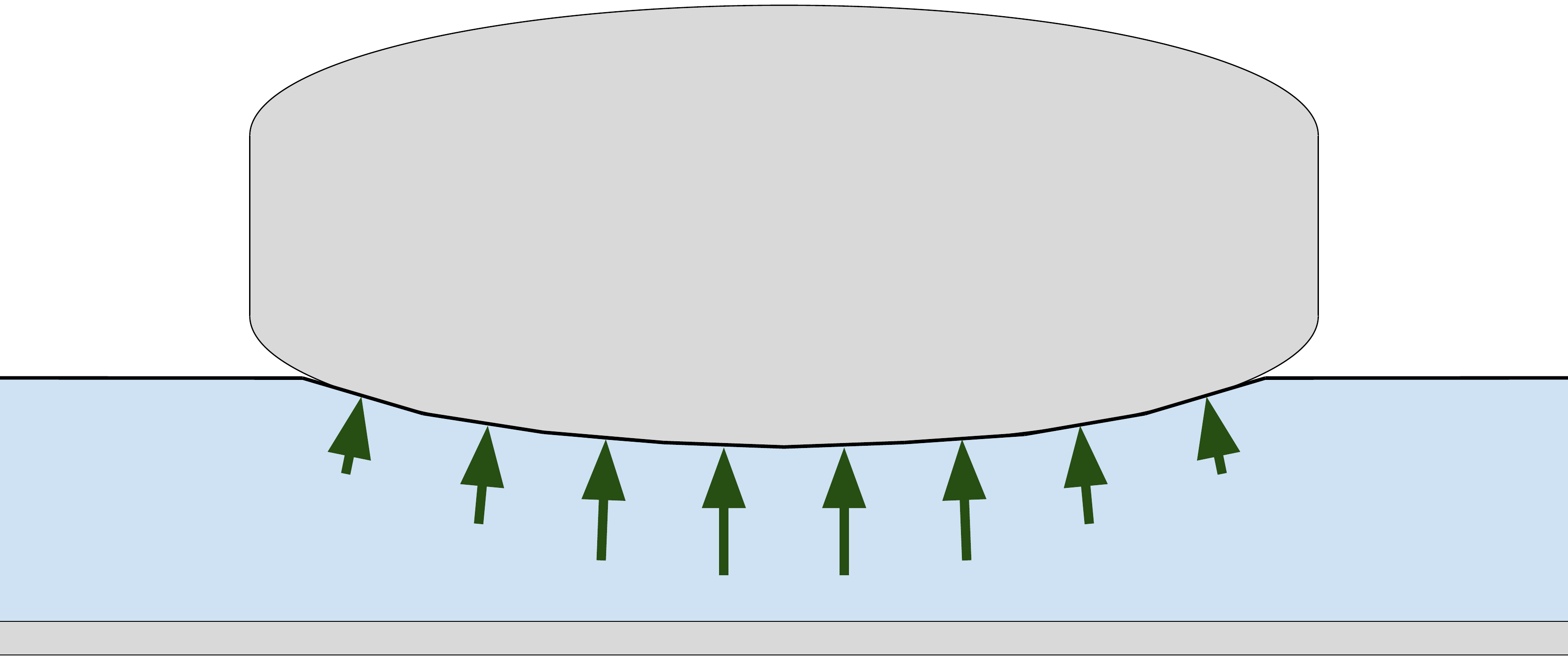}
\includegraphics[width=0.49\linewidth]{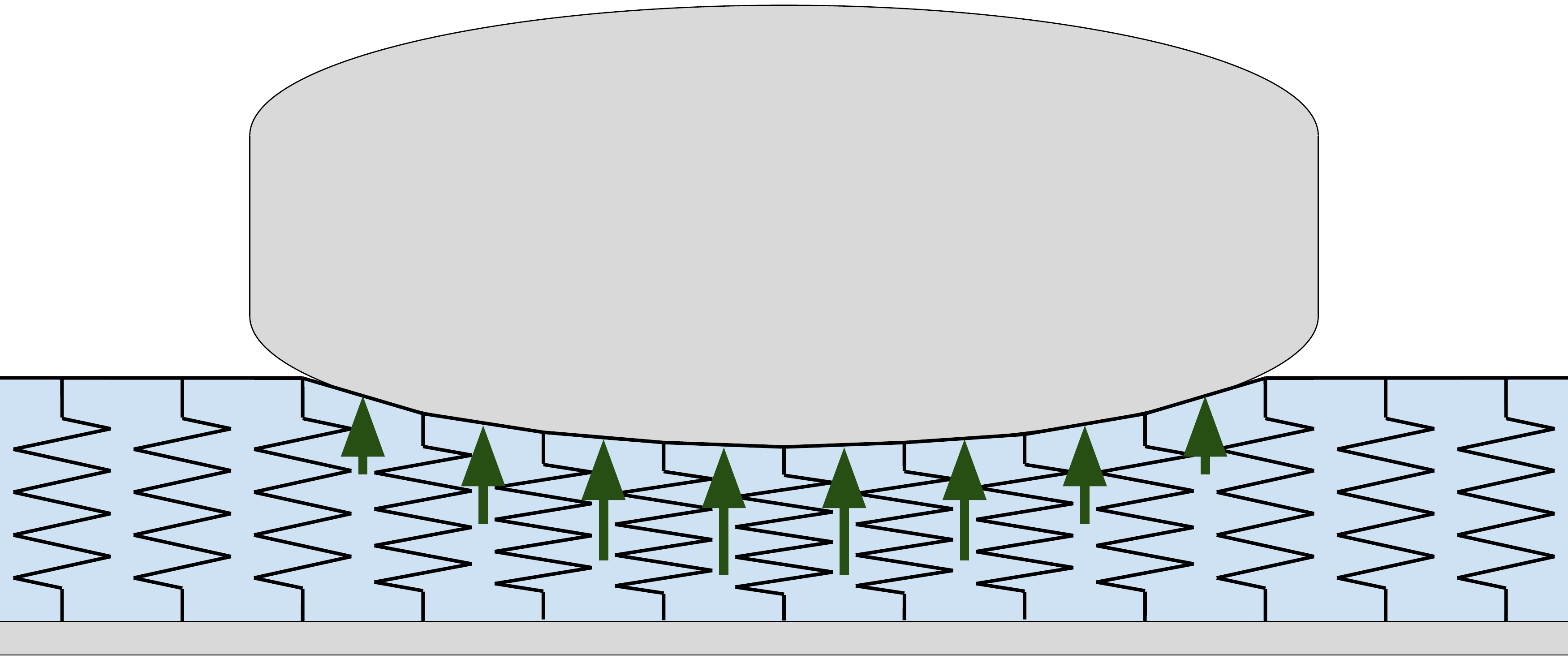}
\caption{\footnotesize\small
    \textbf{Left:}
    Hydrostatics applies a normal traction (\ie pressure) to the penetrating surface.
    \textbf{Right:}
    Elastic foundation models push objects away from an elastic layer.
    This model is not well defined for interpenetrating surfaces with large relative angles.
    \textbf{Both:}
    These models become identical as the penetrating surface becomes flatter.
    In both cases, the traction magnitude is proportional to a penetration distance.
}
\label{fig:hydro_vs_spring}
\end{figure}

\par
\textbf{Relation to overlap volumes.}
Overlap volume models are not explicitly mechanical in their contact wrench calculation.
Instead, a force proportional to volume overlap is applied in some kind of average normal direction, at an average contact location.
The methods do not seem sensible, or even well defined, in  situations with complex geometry.
Where the surfaces are only gently curved, the overlap, elastic foundation, and hydrostatic models seem to be generally equivalent, and are all special cases of PFC.

\section{Adding dissipation}

To model dissipative contact, we add
a term to the pressure that depends on the approach rate of the two bodies at any given point.
Specifically, a general damping function $f$ is dependent upon the pressure, pressure gradient, and approach rate at some point $\point{R}$ on the contact surface. The normal pressure at $\point{R}$ is then defined as:
\begin{equation}
p^{\point{R}} = p_0^{\point{R}} + f(p_0^{\point{R}}, \nabla p_0^{\point{R}}, \RDot, \surfacenormal{n}{R}) \label{eqn:pressure_plus_damping}
\end{equation}
where $p_0^{\point{R}}$ and $\nabla p_0^{\point{R}}$ are the pressure and pressure gradient of the conservative pressure-field at $R$, $\RDot$ is the relative velocity of the two bodies at $R$, and $\surfacenormal{n}{R}$ is the normal direction at $R$. For realizing non-negative dissipation $\chi$ with the Hunt and Crossley model~\cite{Hunt:1975}, for example, one might use the formula:

\begin{align}
        f(p_0^{\point{R}}, \nabla p_0^{\point{R}}, \RDot, \surfacenormal{n}{R})
        \equiv
        \chi\,
        p_0^{\point{R}}\,
        \frac{\nabla p_0^{\point{R}}}{E}
        \cdot
        \hspace{-0.7cm}
        \underbrace{(\surfacenormal{n}{R} ( \RDot \cdot \surfacenormal{n}{R}))}_{\text{velocity normal to contact surface}}
        \hspace{-0.6cm}
        .
        \label{eqn:hc_damping}
\end{align}

\par
\textbf{Friction.} The pressure-field model predicts normal tractions at the calculated equal pressure field surface ${\mathcal S}^\cap$.
For Point $\point{R}$, we can calculate the slip velocity by calculating the relative velocities of the two objects at $\point{R}$, and projecting that onto ${\mathcal S}^\cap$.
Then we can use any friction law that uses pressure and slip velocity.  The simplest is to apply a shear traction against the direction of the projected relative velocity, with magnitude $\mu p^{\point{R}}$.

\section{Implementation}
\label{section:implementation}

This method is comprised of both offline components (discussed in \ref{s:compute_pressure_field}) and online components (discussed in \ref{s:implementation_online}).
Our implementation uses a tetrahedral (tet) approximation of geometry.
For example, a sphere might be approximated by an icosahedron (20 faces) using 20 tets that share a common vertex.
Fewer tets generally implies faster simulations.

\subsection{Computing a pressure field} \label{s:compute_pressure_field}
PFC uses fields that are zero on object boundaries and positive within the interior.
We define $p_0$, the pressure field over an object's domain, as the product of elastic modulus and a quantity $\epsilon$ that we call `penetration extent'; for Point $R$, $p_0^{\point{R}} = E \epsilon^{\point{R}}$ (other mappings from $\epsilon$ to $p_0$ are possible). We define penetration extent to be zero at the object surface and unity at points far from the boundary (at the medial axis, for example). Note that defining $p_0$ as a function of $\epsilon$ allows for non-linear pressure fields using only a single layer of tets from the object boundary to its interior. 

\par
Aside from $\epsilon$, $\nabla \epsilon$ may also be necessary for computing dissipation contributions to tractions since $\nabla p_0 = E \nabla \epsilon$ (more on this below). In general, $\nabla \epsilon$ is not continuous over the entire domain of $\epsilon$. As a workaround, PFC can use an approximation to $\nabla \epsilon$ that we denote $\tilde{\nabla \epsilon}$.

\par
An example of how the pressure field is defined in practice follows.

\subsubsection{Example pressure field for a solid cube}
Here we describe how to create $\epsilon$ (hereafter denoted the \emph{penetration extent field}) and $\tilde{\nabla \epsilon}$, for a solid foam cube centered at the origin.
We divide the cube into 12 tets that share a common vertex at the origin.
The exterior of the cube can participate in contact, so $\epsilon=0$ for all vertices on the boundary while $\epsilon = 1$ at the center of the cube. 

Since $\nabla \epsilon$ is undefined at each vertex, we define $\tilde{\nabla \epsilon}$ at each corner vertex to be a unit vector pointing to the cube origin. $\tilde{\nabla \epsilon}$ is defined to be zero at the cube origin as all tets share this vertex. $\epsilon$ can be defined for many simple shapes using symmetry and similar reasoning. 

\subsubsection{Arbitrary pressure fields with Laplace's Equation} \label{s.laplace}

We now discuss a general way to generate $\epsilon$ and $\tilde{\nabla \epsilon}$ for arbitrary geometry.
We use the following PDE, which is known as Laplace's equation:
\begin{equation} \label{eqn:laplace}
\begin{split}
0 =&~  \frac{\partial^2 \epsilon}{\partial x^2} + \frac{\partial^2 \epsilon}{\partial y^2} + \frac{\partial^2 \epsilon}{\partial z^2}.
\end{split}
\end{equation}

This PDE is solved over a Cartesian domain subject to boundary conditions. It's been used to model various physical phenomena, including steady state heat distributions in thermodynamics.
Figure \ref{fig:iso_laplace} shows a solution to Laplace's equation for an arbitrary 2D geometry.
We generate $\epsilon$ by solving Laplace's equation over the interior of an object subject to unit boundary conditions: zero on the boundary and unity far away from the boundary (\eg at the medial axis or at the boundary of a rigid core inside the otherwise compliant object). Solutions to Laplace's Equation are smooth, so $\tilde{\nabla \epsilon} \equiv \nabla \epsilon$ in an object's interior.

\subsection{Calculating contact forces} \label{s:implementation_online}

Calculating contact forces corresponding to a multibody state consists of three steps:
\1 find all tet-tet pairs the contact surface lies in,
\2 discretize the contact surface and calculate quantities needed to compute traction at quadrature points, and 
\3 compute force and moment.  

\subsubsection{Finding the contact surface}

The tet mesh data structure for each object is pre-computed and stores $\epsilon^V$ and $\tilde{\nabla \epsilon}^V$ at each vertex $V$.
A bounding volume hierarchy (BVH) is built on top of this mesh for broad phase collision, and minimizes the number of primitive-level, tet-tet checks.


\subsubsection{Discretizing the contact surface}
The pressure isosurface between pairs of linear tetrahedral elements lies on a plane (see Appendix \ref{appx:tet_tet_intersect_proof}). If not empty, the intersection between the tets and this plane yields a polygonal piece of the contact surface. 
So solving for the contact surface between each tet-tet pair entails solving only a fast polygon clipping problem.
The contact surface consists of convex polygons from \emph{all} pairs containing part of the isosurface.

\par
We tessellate each of these polygons into triangles and then compute the integral in \eqref{eqn:int_for_f} over each triangle with quadrature.
Quadrature rules approximate this integral by sampling the contact wrenches at particular \emph{quadrature points} within each triangle. 
\par
Solving initial value problems and optimizing trajectories are more efficient when derivatives are continuous function of state.
For PFC, this requirement entails making quadrature points continuous functions of state.
PFC satisfies this requirement deterministically and without tracking quadrature locations by tessellating the intersection polygons in the following way: an $n$ sided polygon is divided into $n$ triangles that share a vertex at the polygon's centroid (see Figure~\ref{f.quad_sub_edges}).
This ensures that only zero-area triangles are added/removed as objects move, so that contact surface discretization changes do not introduce force discontinuities.
Algorithm~\ref{alg:compute-contact-surface} describes how to compute and tessellate the contact surface.

\subsubsection{Integrating tractions}
\label{section:integrating-tractions}


`Traction', denoted $\traction{T}$, is the limit of the forces acting at a cross-section of material as the area goes to zero.
It is a 3D vector with units of pressure. 
Quantities needed to compute traction for bodies $A$ and $B$ are:

\begin{description}
\item [$\point{R}$] a point on the contact surface (in $\mathbb{R}^3$).
\item [$\point{R}_A$] a point on Body $A$ instantaneously coincident with $\point{R}$.
\item [$\point{R}_B$] a point on Body $B$ instantaneously coincident with $\point{R}$.
\item [$\RDot$] the relative Cartesian velocity of the objects at $\point{R}$ (i.e., $\RDot = \dot{\point{R}}_A - \dot{\point{R}}_B$).
\item [$\surfacenormal{n}{R}$] the unit vector normal to the contact surface at $\point{R}$ and pointing from Body $B$ to Body $A$.
\end{description}

The normal traction $\traction{T}^R_{\hat{n}}$ at $\point{R}$ is given by $\traction{T}^R_{\hat{n}} = p^R \surfacenormal{n}{R}$, where $p^R$ is computed using \eqref{eqn:pressure_plus_damping}. Let $\RDotNHat$ and $\RDotPerp$ denote the normal and tangential velocities respectively (i.e., relative velocity between $A$ and $B$ in the plane normal to $\surfacenormal{n}{R}$) where:
\begin{align}
\RDotNHat & \equiv
(\RDot \cdot \surfacenormal{n}{R}) \surfacenormal{n}{R}  \\
\RDotPerp & \equiv \RDot
-
\RDotNHat
\end{align}

Adding Coulomb friction is straightforward (this model must  be `regularized' \cite{Ruina:2015} for continuity, not shown):
\begin{align}
\traction{T}_F^{\point{R}} = \begin{cases}
-\mu p^{\point{R}} \frac{\RDotPerp }{ || \RDotPerp|| } & \textrm{ if } || \RDotPerp|| > 0, \\
0 \in \mathbb{R}^3 & \textrm{ if } || \RDotPerp|| = 0.
\end{cases}
\end{align}

The total traction at $\point{R}$ is the sum of the normal and frictional contributions:
$\traction{T}^{\point{R}} = \traction{T}_{\hat{n}}^{\point{R}} + \traction{T}_F^{\point{R}}$.
$\traction{T}^{\point{R}}$ induces a moment $\point{R} \times \traction{T}^{\point{R}}$ on the bodies.
The net wrench $\wrench{f}$ acting on contacting objects $A$ and $B$ is the integral of tractions over the entire contact surface $\mathcal{S}^{\cap}$:
\begin{align}\label{eqn:int_for_f}
    \wrench{f} = \int \begin{bmatrix}
    R \times \traction{T}^{\point{R}} \\
    \traction{T}^{\point{R}} \\
    \end{bmatrix} d \mathcal{S}^{\cap}.  
\end{align}
Note that the integral of a pressure (in, \eg $N/m^2$) over a surface area (in, \eg $m^2$) is a force (in $N$ in the case of the units we used for pressure and area), as expected.

\begin{figure}[t!]
\centering
\includegraphics[width=.7\linewidth]{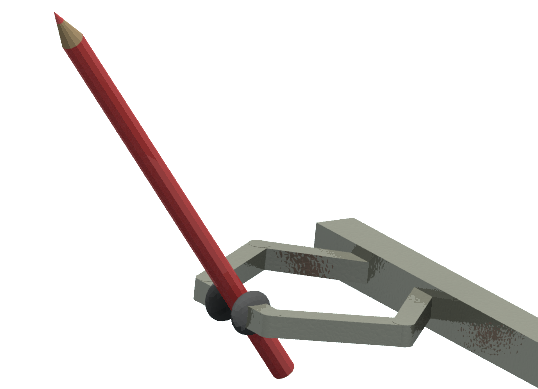}
\caption{\footnotesize\small
Demonstration of a gripper manipulating a pencil in the context of robot controller design using simulation. In this task, the gripper allows the pencil to rotate without falling by loosening the pinch grasp; tightening the grasp causes the rotation to stop. Most point-contact models fail to model this task correctly.
}
\label{f:pencil}
\end{figure}


    

\begin{algorithm*}
\caption{ComputeContactSurface() computes the contact surface between two objects, $A$ and $B$.} \label{alg:compute-contact-surface}
\begin{algorithmic}[1]
\Function{ComputeContactSurface}{$A$, $B$}
\State $\mathcal{S} \leftarrow \emptyset$ \Comment Initialize the contact surface
\State $T \leftarrow$ pairs of tetrahedra from $A$ and $B$ flagged for possible intersection in a broad phase
\ForAll{$T_{A_i}, T_{B_i} \in T$}
  \State $\epsilon_{A_i} \leftarrow$ penetration extent field from $T_{A_i}$
  \State $\epsilon_{B_i} \leftarrow$ penetration extent field from $T_{B_i}$
  \State $\rho \leftarrow$ equilibrium pressure plane from \eqref{eqn:pressure-equilibrium} using $\epsilon_{A_i}$ and $\epsilon_{B_i}$ and Young's Moduli $E_A$ and $E_B$
  \State $P \leftarrow T_{A_i} \cap T_{B_i} \cap \rho$ \Comment{Compute the intersecting polygon}
  \If{$P \neq \emptyset$}
  \State Create a vertex $v_c$ in the centroid of $P$
  \State Tessellate $P$ into set of triangles $\mathcal{T}$ by connecting all edges of $P$ to $v_c$
  \State $\mathcal{S} \leftarrow \mathcal{S} \cup \mathcal{T}$ \Comment Update the contact surface
  \EndIf
\EndFor
\State \textbf{return} $\mathcal{S}$
\EndFunction
\end{algorithmic}
\end{algorithm*}

\begin{figure*}[t]
\centering
\includegraphics[width=1.0\textwidth]{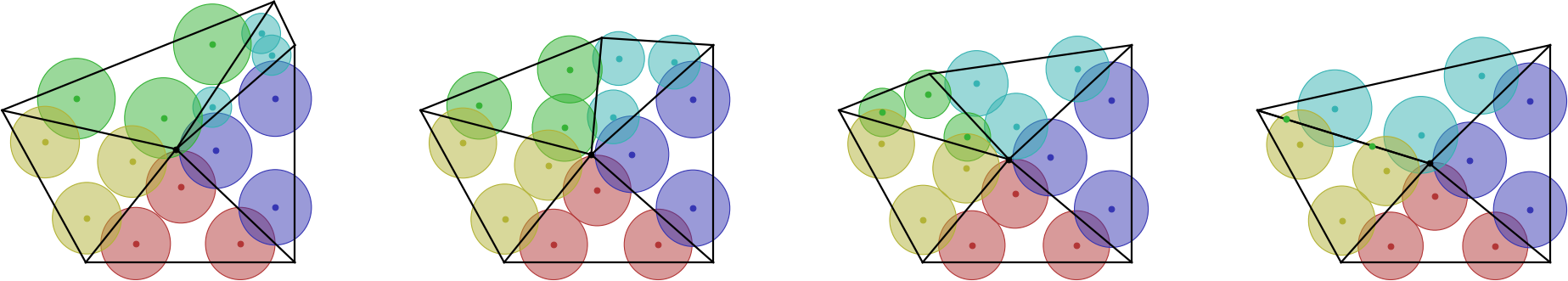}
\caption{\footnotesize\small
    An intersection polygon loses one of its five sides (green circles, left to right).
    The polygon is divided into triangles that share a vertex at the intersection polygon's centroid.
    Small circles represent the position of quadrature points.
    Large circles represent the weight of quadrature points.
    The weight and position of quadrature points move continuously even as an edge disappears.
}
\label{f.quad_sub_edges}
\end{figure*}


\section{Asymptotic Time Complexity Analysis}
The time complexity of this contact model follows from the cost of recursively comparing two bounding volume hierarchies; other operations (e.g., polygon clipping) scale linearly with the number of intersecting tets.
Suppose that $A$ has $n_A$ tets and $B$ has $n_B$ tets, of which $m$ intersect.
The expected time complexity of collision detection is $O(m \lg (n_A + n_B))$ for a BVH\footnote{The depth of a balanced tree with $n_A$ leaves is $O(\lg n_A)$.
Therefore, $O(\lg(n_A + n_B))$ intersection tests need to be performed against each object.}.
In the worst case, $m$ is $O(n_A n_B)$ and corresponds to intersecting combs (i.e., a Chazelle polyhedron) as Fig~1.4 of~\cite{Nguyen:2006} shows.
In practice, recursive hierarchy traversal is usually small; most meshes are not degenerate in this way.

\section{Experiment and demonstrations}
This section describes an experiment and some demonstrations that we conducted using multibody dynamics simulation.
Our simulation was built in Julia and used the RigidBodyDynamics library~\cite{rigidbodydynamicsjl} for dynamics and a custom implementation of the single-step variable order implicit integrator RADAU Hairer describes in \cite{Hairer:1999}.

\subsection{Experiment against FEM} \label{section:experiments}

\par 
We checked the agreement of the quasi-static (i.e., pressure field) aspect of PFC by comparing the normal force produced by PFC against that produced by finite strain elasticity for a periodic 2D plane scenario.
Finite strain elasticity was simulated in \texttt{ANSYS}.
For this scenario, a rigid sinusoidal surface with amplitude $\eta$ and wavelength $\lambda$ was pressed a distance $d$ into a flat, frictionless, compliant layer of unit thickness and Poisson's ratio $\nu = 0.3$.\footnote{$\nu = 0.3$ is typical for many materials including certain foams.}
Figure~\ref{fig:sin_geo} shows this sinusoid and compliant layer when $\lambda$ = $2 \pi / 3$ for three amplitude values $\eta \in [0.0, 0.166, 0.333]$.
The normal force produced by FEM and the pressure field are identical when $\eta = 0$ (conforming contact), as seen in Figure~\ref{fig:sin_geo}.
At the largest amplitude, when $\eta = 0.333$ (point-like contact), the pressure field method underestimates the normal force by 23\%. This experiment shows that PFC agrees best with FEM for conforming contacts (common in grasping), and underestimates forces for point contacts, as expected from PFC's neglect of Poisson effects.

\subsection{Demonstrations}
As shown in the accompanying video, we have tested our approach on a number of manipulation scenarios --- chopsticks manipulating small, soft objects, a robot hand picking a mug, bolt threading in a hole, and two boxes manipulating a spoon --- showing that we can simulate both rigid and compliant objects with convex and non-convex geometries on challenging tasks.
We have tested the approach on contacting boxes and a ball bouncing on a surface (which demonstrates how we can visualize fields like pressure) defined over the contact surface. And we have devised demonstrations that use simulations with the pressure field contact model to design controllers for robots; Figure~\ref{f:pencil} shows one such scenario in which we leverage the scrubbing torques that the model can predict --- and which point based contacts struggle to simulate --- to control the rotation of a pencil within a grasp.

\section{Discussion} \label{s295}
\label{section:discussion}

PFC does not capture certain features of elastic materials including anisotropy, Poisson's effects, and shear effects.
Anisotropic materials like wood do not deform in the direction of applied loads because they are stiffer `along the grain'.
Poisson effects describe how materials deform in directions perpendicular to applied loads (e.g., a rubber band stretched in one direction gets thinner in the other two directions).
When a rigid sphere penetrates a half-plane, the surface of the half-plane near the contact region stretches due to shear effects.
As PFC does not predict these effects, it is most accurate when contact is conforming, like how fingers conform to surfaces during grasping.

\par
We are targeting contact applications that require physical realism but do not require highly accurate predictions of contact deformation. That is, the model must behave like a \emph{physical} system (resists penetration and sliding, produces continuous forces, is energetically consistent, etc.) but does not have to perfectly predict any particular experiment.
We believe this approximation is sufficiently accurate for predicting multi-object contact dynamics, because \1 feedback controllers (in robotics applications) are expected to be able to mitigate the effects of unmodeled disturbances, \2 the initial conditions (particularly deformations) can not be known to sufficient accuracy to get the benefits of increased contact modeling accuracy, and \3 speed can be more important than accuracy in many applications, such as learning control policies.

\par
An open source C++ implementation of this model is available in Drake (\url{https://drake.mit.edu}).

\begin{figure}[t]
    \centering
    \includegraphics[width=0.9825\linewidth]{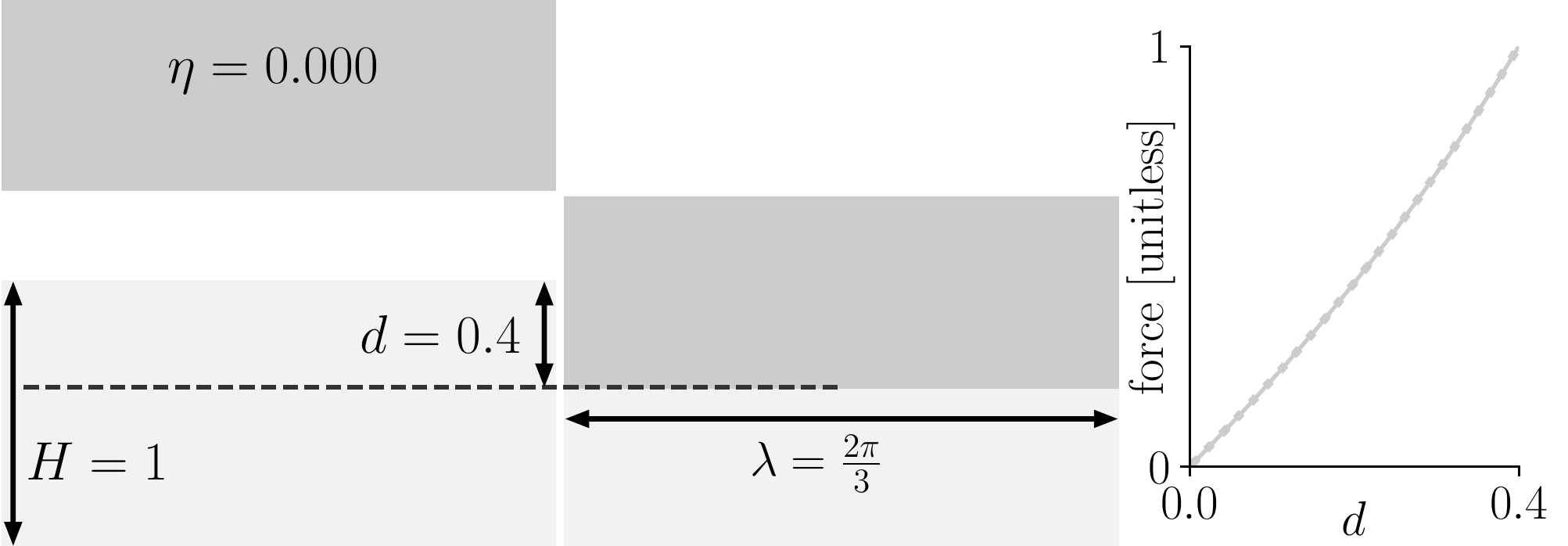}\\
    ~\vspace{-0.3cm}\\
    \includegraphics[width=0.9825\linewidth]{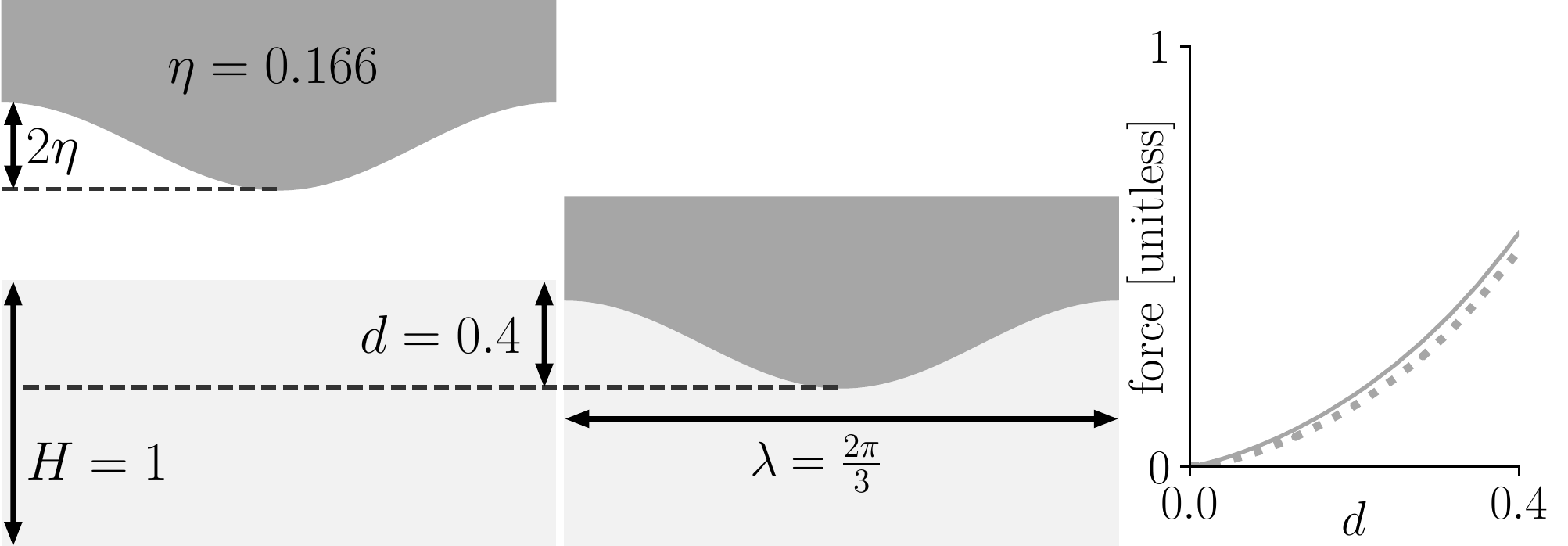}\\
    ~\vspace{-0.3cm}\\
    \includegraphics[width=0.9825\linewidth]{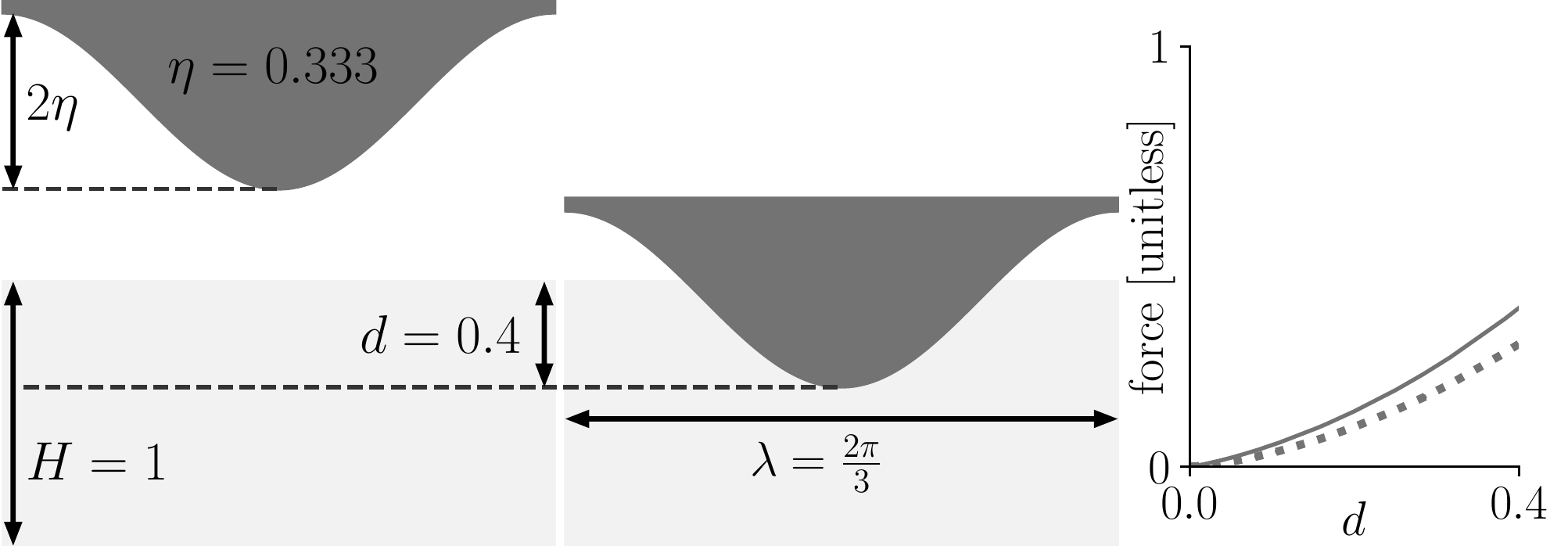}
    \caption{\footnotesize\small
        
    	\textbf{Left:}
    	Sinusoids of varying amplitude above a unit thickness layer.
    	
    	\textbf{Middle:}
    	Sinusoids pressed a distance $d = 0.4$ into the layer. 
    	
    	\textbf{Right:} 
    	Normalized forces predicted by FEM (solid) and PFC (dotted) for each sinusoid.
    	The normalized force is 1 when $\eta=0$ and $d=0.4$.
        
        FEM calculations accounted for large deformations and took all calculations to convergence.

    }
    \label{fig:sin_geo}
\end{figure}

\section{Appendix}
The appendix consists of proofs for the relationship between the static pressure function and potential energy (\S\ref{s:pot_energy_proof}), the location of the equilibrium pressure isosurface for pressure fields defined on linear tets (\S\ref{appx:tet_tet_intersect_proof}), and for the continuity of forces determined using non-continuous (discrete) geometry (\S\ref{appendix:force-continuity}).

\subsection{The potential energy of a pressure field contact} \label{s:pot_energy_proof}

\begin{thm}
The formula:
\begin{equation*}
\begin{split}
U_{p_0} = \int_{V^\cap} p_0(x,y,z) ~dx ~dy ~dz,
\end{split}
\end{equation*}
(aka. Eqn \eqref{e:int_P_dV}) gives the deformation energy.
\end{thm}


\begin{proof}
	Consider the integral of the pressure field $p_0$ over a volume $V^{\cap}(t)$.
	Reynolds transport theorem says that the time derivative of this integral is the sum of the flux of $p_0$ through the volume's boundary $A^{\cap}(t)$ and the volume integral of $\dot{p}_0$.
	
\begin{align} \label{e:rtt}
\frac{d}{dt} \int p_0 ~d V^\cap
=~
\underbrace{
\int \dot{p}_0 ~d V^\cap}_{\text{change due to~} \dot{p}_0 } +
\underbrace{\int (\velocity \cdot \nhat) p_0 ~d A^\cap}_{\text{change due to boundary}}
	\end{align}
	Without loss of generality, we perform calculations in an object-fixed frame to neglect the ``change due to $\dot{p}_0$'' term.
	In this frame, the term associated with the ``change on the boundary'' is the power delivered by elastic forces.
	\begin{align} \label{e:rtt_body_fix}
	\frac{d}{dt} \int p_0 ~d V^\cap
	=~
	\hspace{-0.6cm}
	\underbrace{\int \velocity \cdot p_0 \nhat ~d A^\cap}_{\text{power delivered by elastic forces}} 
	\hspace{-0.6cm}
	\end{align}
	In the absence of dissipation, if we take the time integral of~\eqref{e:rtt_body_fix} starting when the intersection volume is zero (i.e. $V^\cap = \emptyset$) we get \eqref{e:int_P_dV}.
	In other words, the potential energy is the work needed to displace a pressure field.

\end{proof}

\subsection{Linear isoparametric tetrahedron intersection proof}
\label{appx:tet_tet_intersect_proof}

Assume tet $A$ is defined by vertices $\point{\nu}_{A_1}, \ldots, \point{\nu}_{A_4} \in~\mathbb{R}^3$.
Points inside $A$ can be written as a linear combination of these four vertices; the scaling factors (barycentric coordinates), $\point{\zeta} \in \mathbb{R}^4$, are non-negative and sum to one.
A $4 \times 4$ matrix denoted by $\transform{X}^A$, transforms points from $A$ to a Cartesian frame, e.g.:
\begin{align} \label{eqn:tet_vert_transform}
\begin{bmatrix} R \\ 1 \end{bmatrix} =\ 
\begin{bmatrix}
\point{\nu}_{A_1} & \point{\nu}_{A_2} & \point{\nu}_{A_3} & \point{\nu}_{A_4} \\
1 & 1 & 1 & 1
\end{bmatrix}
\point{\zeta}^A
=\ \!\transform{X}^A \point{\zeta}^A
\end{align}
where the subscript `$A$' ($_A$) denotes a quantity relevant to tet $A$ and the superscript `$A$' ($^A$) denotes a point in $A$'s barycentric coordinates. 

\label{section:tetrahedral-proof}
\begin{thm}
For linear pressure fields defined on compliant bodies, points of equal pressure for two tets lie on a plane.
\end{thm}

\begin{proof}

\noindent
The static pressure $p_0$ at Point $R$ is $p_{0}^{\point{R}} = E_A \epsilon^{\point{R}}$, where $E_A \in \mathbb{R}$ is the Young's Modulus of object $A$.
Let $\boldsymbol{\epsilon}_A$ denote the values of $\epsilon$ at $A$'s vertices.
Cartesian position and $\epsilon$ are both dot products of $\point{\zeta}^A$ and the vertex values, therefore $\epsilon = \boldsymbol{\epsilon}_A \cdot  \point{\zeta}^A$.
Combining these relationships with \eqref{eqn:tet_vert_transform}:

\begin{align}
p_{0_A} =  E_A ~\boldsymbol{\epsilon}_A  \cdot \underbrace{ \bigg({}^A\!\transform{X} \begin{bmatrix} \point{R} \\ 1 \end{bmatrix} \bigg)}_{\point{\zeta}^A}
\end{align}

The quasi-static assumption says that the contact surface is described by the set of points where the pressures between two intersecting tets are equal, i.e.:
\begin{align}
p_{0_A}^{\point{R}} = p_{0_B}^{\point{R}},\ \forall \point{R} \in \mathbb{R}^3 \label{eqn:equal-pressures}
\end{align}

Put another way, Point $\point{R}$ in Cartesian space is on the contact surface if \eqref{eqn:equal-pressures} is satisfied.

\begin{equation}
0 =
\ \underbrace{E_A  \boldsymbol{\epsilon}_A \cdot  \bigg({}^A\!\transform{X} \begin{bmatrix} \point{R} \\ 1 \end{bmatrix} \bigg)}_{p_{0_A}^{\point{R}}}
- 
\underbrace{E_B \boldsymbol{\epsilon}_B \cdot \bigg( {}^B\!\transform{X} \begin{bmatrix} \point{R} \\ 1 \end{bmatrix} \bigg) }_{p_{0_B}^{\point{R}}} 
\end{equation}

After rearranging and collecting terms:
\begin{align} \label{eqn:pressure-equilibrium}
0 &=
\bigg( \frac{\boldsymbol{\epsilon}_A}{E_B} \cdot {}^A\!\transform{X} -
\frac{\boldsymbol{\epsilon}_B}{E_A} \cdot {}^B\!\transform{X} \bigg)  
\begin{bmatrix} R \\ 1 \end{bmatrix}
\end{align}

The expression in parentheses is a $1 \times 4$ matrix, therefore, scaling \eqref{eqn:pressure-equilibrium} by some constant produces the canonical equation for a plane: $0 = \hat{n} \cdot R + d$.
Points of equal pressure for two tets lie on plane defined by $\hat{n}$ and $d$ that lies within the intersection of the tets.
\end{proof}

\subsection{Continuous forces from discrete geometry}
\label{appendix:force-continuity}

Some contact models such as those that combine minimum translational distance and point-contact produce forces that are not continuous functions of state (see, e.g.,~\cite{Williams:2014}). 
This section will show that the normal force produced by the pressure field model is a continuous function of state even in the presence of damping, discretization, and quadrature.
We denote the contact surface $\mathcal{S}^{\cap}$, and assume that this surface does not touch a rigid boundary of either object.

\par
Assuming the dissipation function in \eqref{eqn:hc_damping}, force from contact over the contact surface is:

\begin{align} \label{eqn:int_p_n_hat_S}
\wrench{f} =& \int p_0 \big(1 + \chi\, 
\hspace{-0.2cm}
\underbrace{\nabla \epsilon}_{ \nabla p_0^{\point{R}} / E } 
\hspace{-0.2cm}
\cdot 
\hspace{-0.4cm}
\overbrace{ \RDotNHat }^{\surfacenormal{n}{R} ( \RDot \cdot \surfacenormal{n}{R})}
\hspace{-0.4cm}
\,
\big) \surfacenormal{n}{R} ~d \mathcal{S}^{\cap}.
\end{align}


$p_0$ is a continuous function of state in the interior of each object as it is described by a weighted average of tet vertex values. Similarly, $\hat{n}$ (defined in \S\ref{section:integrating-tractions}) is continuous since it is the result of a normalization operation applied to a difference between two weighted averages of tet vertex values.

\subsubsection{At every instant, the contact surface is contiguous in the interior}

Tet $A_1$ belongs to mesh $A$.
Tets $B_1$ and $B_2$ belong to mesh $B$ and share a common face.
Let $P_{A_1 B_1}$ and $P_{A_1 B_2}$ denote the planes associated with the intersections of $A_1$ with $B_1$ and $B_2$ respectively.
$P_{A_1 B_1}$ and $P_{A_1 B_2}$ share the same intersection (if any) with the face common to $B_1$ and $B_2$.
As this coincident intersection is true in general, the contact surface is contiguous in the interior because the internal edges of clipped planes are shared.

\subsubsection{The contact surface is a continuous function of state}
As shown in \S\ref{section:tetrahedral-proof}, tet/tet pressure isosurfaces lie on planes.
The intersection plane associated with any pair of tets is a continuous function of state because \eqref{eqn:pressure-equilibrium} is differentiable.
It follows that $\hat{n}$ is a continuous function of state.
The contact surface is spatially continuous in the interior at all times and consists of planes whose movement is a continuous function of state.
It follows that the contact surface is a continuous function of state, for a given tessellation of the contact surface.

\subsubsection{Changes in contact surface mesh do not introduce discontinuities}
It remains to be shown that discrete changes to the contact surface mesh do not introduce discontinuities.
Figure \ref{f.quad_sub_edges} illustrates the argument here.
Because of the continuity established above, triangle edges move continuously with state.
They must therefore reach zero area just prior to being removed from the mesh, and new edges must enter in such a way that they initially form zero-area triangles.
Because our forces result from integrals over these triangular faces, the force contributions of new and deleted triangles are zero, so no discontinuity occurs.
Subsequently they evolve continuously as established above.

\bibliographystyle{abbrv}
\bibliography{paper}
	
\end{document}